\documentclass[runningheads]{llncs}
\usepackage{latexsym}
\usepackage{amssymb}
\usepackage{amsmath}
\usepackage{amsfonts}
\usepackage{color}
\usepackage{epsfig}
\usepackage{graphicx}
\usepackage{epstopdf}
\usepackage{fullpage}
\newtheorem{obs}[theorem]{Observation}
\makeatletter
\newcommand{\oset}[2]{%
  {\mathop{#2}\limits^{\vbox to -.9\ex@{\kern-\tw@\ex@
   \hbox{\scriptsize #1}\vss}}}}
\makeatother
\begin{document}
\title{\large $d$-COS-R is FPT via Interval Deletion}
\author{N.S. Narayanaswamy \and R. Subashini}
\institute{Department of Computer Science and Engineering\\ Indian Institute of Technology Madras, India.\\ 
\email{\{swamy,rsuba\}@cse.iitm.ac.in}}
\maketitle
\begin{abstract}
\noindent 
A binary matrix $M$ has the Consecutive Ones Property (COP) if there exists a permutation of columns that arranges the ones consecutively in all the rows. Given a matrix, the $d$-COS-R problem is to determine if there exists a set of at most $d$ rows whose deletion results in a matrix with COP. We consider the parameterized complexity of this problem with respect to the number $d$ of rows to be deleted as the parameter. The closely related Interval Deletion problem has recently shown to be FPT \cite{CM12}. In this work, we describe a recursive depth-bounded search tree algorithm in which the problems at the leaf-level are solved as instances of Interval Deletion.  The running time of the algorithm is dominated by the running time of Interval Deletion, and therefore we show that $d$-COS-R is fixed-parameter tractable and has a run-time of  $O^*(10^d)$. 
%{\bf Keywords:}  Consecutive Ones Property, Consecutive Ones Submatrix, Parameterized complexity
\end{abstract}
\section{Introduction}
\label{sec:intro}
%It is known from \cite{HG02} that d-COS-R on (2,2)-matrices is polynomial-time solvable. Surprisingly, the closely related variants d-COS-R on (2,3)  and (3,2)-matrices are NP-hard \cite{TZ07}.
%{\bf Related work}:
\par Testing  COP for binary matrices is a classical algorithmic problem.  COP testing has applications in physical mapping of DNA \cite{ABH98} and in recognizing interval graphs, planar graphs and Hamiltonian cubic graphs \cite{BL76,WLZ07}. There are many linear-time algorithms known in the literature for COP testing \cite{BL76,HMPV00,Hsu02,Hsu_Mc,MPT98,TM05}.  There are many combinatorial properties of matrices with COP.  They are known to be totally unimodular, and there are results connecting matrices with COP and Intersection Cardinality Preserving Interval assignments \cite{FG65,NS09}. Further, the classical NP-hard problems, integer linear programming (ILP) and set cover, are polynomial-time solvable, when the associated binary matrix has COP \cite{MD08}. 
\noindent
In this paper our focus is on matrices that do not have COP, and we address the natural optimization problem to find a minimum set of rows whose deletion results in a submatrix with COP. The corresponding decision problem, referred to as $d$-COS-R, is known to be NP-complete \cite{GJ79} and is well-studied in the parameterized complexity framework \cite{DGN10}. A parameterized problem is said to be fixed-parameter tractable (FPT) with respect to $d$ as the parameter if there is an algorithm with run-time $O^*(f(d))$, where $f$ is a computable function depending only on $d$ \footnote{$O^*$ notation ignores polynomial terms.}. For details on parameterized complexity, we refer the reader to \cite{DowneyF99,Nie2006}. In this paper, we consider the parameterized complexity of $d$-COS-R defined as follows: \\

\vspace{-.25cm}
\noindent \fbox{
\parbox{16cm}{
%\parbox{11.7cm}{
\hspace{7cm}{\bf $d$-COS-R}   \\                         
\textit{\textbf{Instance:}} $(M,d)$ - A binary matrix $M_{m \times n}$ and an integer $d \ge 1$. \\
\textit{\textbf{Parameter:}} $d$ \\
\textit{\textbf{Question:}} Does there exist a set of at most $d$ rows of $M$ whose deletion results in a matrix with COP? 
}
} \\
\noindent
The problems of deleting a minimum number of rows or columns to transform a given matrix into a matrix with COP are called Min-COS-R and Min-COS-C, respectively.  These two problems are known to be NP-hard even on very sparse matrices, containing only two 1-entries per row and at most three 1-entries per column \cite{TZ07}.  These minimization and the corresponding maximization versions have been studied \cite{DGN10}.  Min-COS-R and Min-COS-C are fixed-parameter tractable on matrices that have only two ones either per row or per column.   In this work we focus only on the decision version of Min-COS-R which is the $d$-COS-R problem.  On restricted classes of matrices, $d$-COS-R is known to be FPT \cite{DGN10}. 
%The notation {\em $(x,y)$-matrix} denotes a matrix consisting of at most $x$ number of ones in each column and at most $y$ number of ones in each row. On $(2,*)$-matrices and $(*,2)$-matrices, $d$-COS-R is shown to be FPT.  Also, $d$-COS-R on $(*,\Delta)$-matrices is FPT (w.r.t $d$ and $\Delta$ as parameters) with run-time $O^*((\Delta+1)^d.(2\Delta)^{2.min\{d,4 \Delta^2\}})$. 
These FPT algorithms are based on a refinement of the forbidden submatrix characterization of matrices with COP \cite{AT72}.   To the best of our knowledge, the parameterized complexity of $d$-COS-R on general binary matrices is still open. In this work, we show that $d$-COS-R admits an algorithm with run-time $O^*(10^d)$.  Our result is obtained by a recursive branching algorithm in which the leaf instances are that of Interval-Deletion (defined below).   Then, we employ the recent $O^*(10^d)$ algorithm for Interval Deletion \cite{CM12} to solve $d$-COS-R. Thus, we answer the natural open question on the parameterized complexity of $d$-COS-R by showing that it is FPT on all binary matrices.  This is a significant advancement over the current knowledge on this problem, where current FPT results \cite{DGN10} are known only when there are bounds on the number of 1s in the rows or columns.  

%\vspace{-.25cm}
\noindent \fbox{
\parbox{16cm}{
%\parbox{11.7cm}{
\hspace{7cm}{\bf Interval Deletion}   \\                         
\textit{\textbf{Instance:}} $(G,d)$ - A graph $G$ and an integer $d \ge 1$. \\
\textit{\textbf{Parameter:}} $d$ \\
\textit{\textbf{Question:}} Does $G$ have a set $V'$ of at most $d$ vertices such that $G \setminus V'$ is an interval graph? 
}
} \\ \\
\noindent 
{\bf Our Approach:} A natural approach towards obtaining submatrices with COP is to identify the known classes of forbidden configurations \cite{AT72}, and to remove them by eliminating appropriate rows.  While this is the broad approach in \cite{DGN10}, we look at the well known fact that a graph is an interval graph if and only if its {\em clique matrix} (formally defined later) has COP \cite{FG65,Gol04}. We consider the question of how to convert a given 0-1 matrix into the clique matrix of some graph, and then attempt an interval deletion on that graph. From \cite{Gol04}, a natural graph that can be associated with a binary matrix is a {\em derived graph}.  Informally, the columns of the matrix correspond to cliques in the derived graph. However, a derived graph may have many other {\em spurious} cliques, and these cliques are the hinderances towards getting a clique matrix. Our first branching rule motivated by the Helly property, that must be satisfied by any set of intervals, ensures that these spurious cliques are localized to the derived graph associated with a pair of columns in the given matrix.  We then design a second branching rule, based on induced 4-cycles,  to ensure that the number of these spurious cliques is a polynomial in the input size, and they can be enumerated in polynomial time.  Then with a third set of branching rules we eliminate these spurious cliques, the result being a matrix in which the maximal cliques of the derived graph are associated with some column of the matrix.  We then consider an {\em augmented matrix} which becomes the clique matrix of a graph.  We then show that Interval Deletion on this graph ensures that the augmented matrix has COP, which directly gives a submatrix with COP for the given matrix.  All these branching rules, along with the recent FPT algorithm \cite{CM12} for Interval Deletion are shown to solve the $d$-COS-R problem in FPT time.

\section{COP, Intervals, and Clique-Matrices}  \label{sec:climat}
In this section, we present the necessary structural results to describe our algorithm and the proofs of correctnesses.   Some of the lemmas are cited from the appropriate papers, and some are proved by us.   Graph theoretic definitions and notations are as per \cite{WDB01,Gol04}.

\noindent  Throughout this paper we consider only binary matrices. For an $m \times n$ matrix $M$, let ${\cal R}(M)=\{r_1,\ldots,r_m\}$ and ${\cal C}(M)=\{c_1,\ldots,c_n\}$ denote the sets of rows and columns, respectively. The $(i,j)^{th}$ entry in $M$ is denoted as $M_{ij}$. For a subset ${\cal D} \subseteq {\cal R}(M)$ of rows, the submatrix induced on $\cal D$ and ${\cal R}(M) \setminus {\cal D}$ are denoted by $M[{\cal D}]$ and $M \setminus {\cal D}$, respectively. 

\noindent
The {\em derived graph} associated with a 0-1 matrix $M$, defined in \cite{Gol04}, is $G(M)=(V,E)$  is defined as $V=\{v_i \mid r_i \in {\cal R}(M)\}$ and $E=\{\{v_i,v_j\} \mid \exists c_k\in{\cal C}(M), ~ M_{ik}=M_{jk}=1\}$.  In other words, $G(M)$ is obtained from $M$ by visualizing each column as a clique involving the vertices (corresponding to rows)  which have a 1 entry in that column.  For a column $c_k$ in $M$, the support of $c_k$, denoted by $supp(c_k)$, is defined as the set $\{r_i \in {\cal R}(M) \mid M_{ik}=1\}$. Also, for $c_k$, the set of vertices in $G(M)$ corresponding to the rows in $supp(c_k)$ is defined as $vert(c_k)=\{v_i \mid r_i \in supp(c_k)\}$. 
\subsection{Matrices with COP, Interval Assignments, and Interval Graphs} 
%A graph class that is closely related to the COP of matrices is the class of interval graphs. 
A graph is called an interval graph if its vertices can be assigned intervals such that two vertices are adjacent if and only if their corresponding intervals have nonempty intersection. Let $G$ be a graph on the vertex set $\{v_1,\cdots,v_n\}$ and let $\{Q_1,\cdots,Q_l\}$ be the set of maximal cliques in $G$. The {\em clique matrix} $M$ of $G$ is the matrix whose rows and columns correspond to the vertices and the maximal cliques, respectively, in $G$. The entry $M_{ij}=1$ if the vertex $v_i$ is in the clique $Q_j$ and it is 0 otherwise. 
%The clique matrix of a graph is unique up to row and column permutations \cite{Gol04}. 
The following characterization relates COP and interval graphs.
\begin{theorem}
\label{FG}
{\em \cite{FG65}} A graph is an {\em interval graph} if and only if its clique matrix has COP.
\end{theorem}
\begin{theorem}
\label{char-int}
{\em \cite{Gol04}} A graph $G$ is an {\em interval graph} if and only if $G$ has no induced cycle of length 4 and $\overline{G}$ is a comparability graph.
\end{theorem}
Here we set up the framework to argue the correctness of our branching rules.
%The branching rules employed in our algorithm rely on a characterization of matrices with COP using interval assignments \cite{FG65,NS09}. 
%We first fix the necessary notation required for describing this characterization. 
An $m \times n$ matrix $M$ can be represented as a set system ($U,{\cal S}(M)$) with ${\cal S}(M)$= $\{S_1,\ldots,S_m\}$ being a collection of subsets of $U=\{1,\ldots,n\}$ where $S_i=\{j \mid M_{ij}=1\}$. A family of subsets is said to have the {\em Helly property} if every subfamily of it, formed by pairwise intersecting subsets, contains a common element \cite{DPS09}. An interval $J$, denoted by $[i,k]$, is the ordered set of consecutive integers from $i$ to $k$. An interval assignment ${\cal I}$ to a set system ($U,{\cal S}$) is an assignment of an interval $I_i$ to each $S_i \in {\cal S}$. An Intersection Cardinality Preserving Interval Assignment (ICPIA) to ${\cal S}$ is an interval assignment ${\cal I}$ that satisfies $|S_{i} \cap S_{j}|=|I_{i} \cap I_{j}|$ for every pair $S_i$ and $S_j$ of elements in ${\cal S}$.  A main property of the ICPIA, shown in \cite{FG65,NS09} is that for any collection of sets $\{S_{i_1}, \ldots, S_{i_r}\}$, $\displaystyle |\bigcap_{j=1}^r S_{i_j}| =  |\bigcap_{j=1}^r I_{i_j}|$.  
%It is known that an ICPIA for a set system (if one exists) can be obtained in polynomial time \cite{NS09}. %
\begin{theorem}
\label{NS}
{\em \cite{FG65,NS09}} A matrix $M$ has COP if and only if ${\cal S}(M)$ has an ICPIA.  Further, if $\cal I$ is an ICPIA for ${\cal S}(M)$, then 
for any collection of sets $\{S_{i_1}, \ldots, S_{i_r}\} \subseteq {\cal S}(M)$, $\displaystyle |\bigcap_{j=1}^r S_{i_j}| =  |\bigcap_{j=1}^r I_{i_j}|$.  
\end{theorem}
We prove a key lemma that is necessary for the first rule in our branching algorithm.
\begin{lemma} 
\label{lem-helly}
If $M$ has COP then ${\cal S}(M)$ satisfies the Helly Property.  Further, for every triple of pairwise intersecting sets in ${\cal S}(M)$, one of the sets is contained in the union of the other two.
\end{lemma}
\begin{proof}
Since $M$ has COP, let $M'$ be the column permuted matrix obtained from $M$ which has consecutive ones in the rows.  For each $1 \leq i \leq m$, let $I_i$ be the natural interval assigned to $S_i$, obtained from $M'$.  Let ${\cal I} = \{I_1, \ldots, I_m\}$ be this interval assignment.   From Theorem \ref{NS}, $\cal I$ is an ICPIA for ${\cal S}(M)$.  Therefore, if there exists three sets $S_1, S_2, S_3$ that violate the Helly property- we first observe that for each pair of them, say $S_i$ and $S_j$, $|S_i \cap S_j| = |I_i \cap I_j| > 0$.  Since intervals satisfy the Helly property, it follows that the 3 intervals have a common point.  We now conclude that $0 = |S_1 \cap S_2 \cap S_3| = |I_1 \cap I_2 \cap I_3| > 0$.    The first equality comes from our hypothesis that the 3 sets violate Helly property, the second equality follows from Theorem \ref{NS}, and the third inequality follows from the fact that the 3 intervals share a common point, as
intervals respect Helly Property.  This is a contradiction to our premise that $S_1, S_2, S_3$ violate the Helly Property, which is now shown to be false.  To prove the second part of the lemma, let $S_1, S_2, S_3$ be pairwise intersection sets.  Then, we know that in
the corresponding intervals $I_1, I_2, I_3$, one of them is contained in the union of the other two, say $I_3$ is contained in $I_1 \cup I_2$.  Since $I$ is
an ICPIA, it follows that $S_3 \subseteq S_1 \cup S_2$.  Hence the lemma.\qed
%then the corresponding intervals A given set of intervals satisfy Helly property, Now, for every subset $\cal J$ of pairwise intersecting intervals of ${\cal I}$, %we have $\bigcap_{J \in {\cal J}} \neq \emptyset$.\qed
\end{proof}
%\begin{corollary}
%\label{cor-2}
%If ${\cal I}$ is an interval assignment that satisfies Helly property, then for every triple of pairwise intersecting intervals in ${\cal I}$, one of the intervals is contained in the union of the other two
%\end{corollary}
%\begin{lemma}
%\label{cor1}
%If $M$ has COP then for every triple of pairwise intersecting sets in ${\cal S}(M)$, one of the sets is contained in the union of the other two. 
%\end{lemma} 
\subsection{Matrices with COP and Clique-Matrices of Derived Graphs}
For $M$, the $(n+m) \times n$ matrix $\oset{$\thicksim$}M$ is defined as $\left(\smallmatrix I  \\ M \endsmallmatrix \right)$ where $I$ is the $n \times n$ identity matrix.  The main reason for considering $\oset{$\thicksim$}M$ is that in $G(\oset{$\thicksim$}M)$, each column corresponds to a maximal clique.  This may not necessarily be the case in $G(M)$.   We first observe that $M$ and $\oset{$\thicksim$}M$  behave the same with respect to COP, and the proof of this observation is very easy based on the fact that $\oset{$\thicksim$}M$ is obtained from $M$ by {\em padding} an  identity matrix.
%We now describe some properties of matrices with COP that are crucially used in our algorithm. 
\begin{obs}
\label{COP}
$\oset{$\thicksim$}M$ has COP if and only if $M$ has COP. 
\end{obs}
%\begin{proof}
%If $\oset{$\thicksim$}M$ has COP then for every subset ${\cal D} \subseteq {\cal R}(\oset{$\thicksim$}M)$ of rows, the submatrix $M[{\cal D}]$ has COP. Since $M$ is one such submatrix of $\oset{$\thicksim$}M$, it also has COP. Conversely, if $M$ has COP, then there exists a permutation $\sigma$ of columns of $M$ that witnesses it. By definition, $\oset{$\thicksim$}M$ comprises of the rows of $I$ followed by the rows of $M$. Since $I$ has exactly one 1 in every row, it follows that $\sigma$ is also a permutation of columns of $\oset{$\thicksim$}M$ that results in COP.\qed
%\end{proof}
\begin{corollary}
\label{corr-cop}
Let ${\cal D} \subseteq {\cal R}(M)$. Then, $M \setminus {\cal D}$ has COP if and only if $\oset{$\thicksim$}M \setminus {\cal D}$ has COP.
\end{corollary}
\begin{lemma}
\label{cop-ig}
If $M$ has COP, then $G(M)$ is an interval graph.  Further,  for every maximal clique $Q$ in $G(M)$ there exists a column $c_k$ in $M$ such that $vert(c_k)=Q$.
\end{lemma}
\begin{proof}
Consider the columns of $M$ in the order of a permutation $\sigma$ that results in COP. Now, for every vertex $v_i$ in $G(M)$ assign the interval $I_i=[j,k]$ where $j$ and $k$ are the minimum and maximum column indices, respectively, with $M_{ij}=M_{ik}=1$. Consider two vertices $v_a$ and $v_b$ in $G(M)$. Let $I_a=[j_1,k_1]$ and $I_b=[j_2,k_2]$ be the intervals assigned to $v_a$ and $v_b$ respectively. Now, by the definition of derived graphs, $v_a$ and $v_b$ are adjacent if and only if there is a column $c_r$ ($\mbox{min }\{j_1,j_2\} \leq r \leq \mbox{min }\{k_1,k_2\}$) in $M$ with $M_{ar}=M_{br}=1$. The existence of such a column $c_r$ is well-defined if and only if $I_a \cap I_b \neq \emptyset$. Therefore, $v_a$ and $v_b$ are adjacent in $G(M)$ if and only if $I_a \cap I_b \neq \emptyset$. Thus, $G(M)$ is an interval graph.
\noindent
We now prove the second part of the lemma.  Let $Q=\{v_1,\cdots,v_q\}$ be a maximal clique in $G(M)$. Consider the submatrix $M'$ with ${\cal R}(M')=\{r_i \in {\cal R}(M) \mid v_i \in Q\}$.  Recall that   ${\cal S}(M') = \{S_i \mid r_i \in {\cal R}(M')\}$.  Any two $S_i, S_j \in {\cal S}(M')$ have a non-empty intersection, and therefore, from Lemma \ref{lem-helly} , if follows that $\displaystyle |\bigcap_{i=1}^q S_i | > 0$.   Let $k$ be an element in $\bigcap_{i=1}^q S_i$, then it follows that $vert(c_k) = Q$. Note that $vert(c_k) = Q$ because $Q$ is a maximal clique.  Hence the lemma. \qed  
\end{proof}
\begin{corollary}
\label{IG-1}
If $M$ has COP, then $G(\oset{$\thicksim$}M)$ is an interval graph, and $\oset{$\thicksim$}M$ is the clique matrix of $G(\oset{$\thicksim$}M)$.
\end{corollary}
\begin{proof}
From Observation \ref{COP}, $M$ has COP implies that $\oset{$\thicksim$}M$ has COP, and from Lemma \ref{cop-ig} it follows that $G(\oset{$\thicksim$}M)$ is an interval graph, and that each maximal clique corresponds to a column in $\oset{$\thicksim$}M$. Now in  $\oset{$\thicksim$}M$, each column has a
{\em distinguishing} entry where there is a 1, and all other entries in that row are zero.  This shows that 
 each column corresponds to a maximal clique in $G(\oset{$\thicksim$}M)$.   Therefore, $\oset{$\thicksim$}M$ is the clique matrix of $G(\oset{$\thicksim$}M)$. \qed
\end{proof}
\section{$d$-COS-R via  Interval Deletion}
\label{our-algo}
The basic idea in this algorithm is that we transform the given instance $(M,d)$ of $d$-COS-R to an instance $(M',d')$ where $M'$ has the additional property that $\oset{$\thicksim$}M'$ is the clique matrix of a graph $G(\oset{$\thicksim$}M')$.  Our recursive algorithm explores a recursion tree in which each leaf corresponds to an interval deletion problem.\\
%For a set $r_i$ in ${\cal R}(M)$, the corresponding vertex in $G(M)$ is denoted by $v_i$. 
%Now, we describe a set of branching rules that preprocess $M$ such that for every maximal clique $Q $ in $G(M)$,  there exists a column $c_k$ in $M$ such that $vert(c_k)=Q$. and $|{\cal D}| \leq d$
\noindent \fbox{
\parbox{16cm}{
%\parbox{11.7cm}{
{\bf Algorithm COS-R($M,{\cal D},d$)}\\
Input: An instance ${\cal I}= ( M_{m \times n},d ) $ where $M$ is a binary matrix and $d \ge 1$. \\
Output: Return a set $\cal D$ of at most $d$ rows (if one exists) such that $M \setminus {\cal D}$ has COP.\\
% {\bf Step 1:} The following rules are applied in the order in which they are specified. Any rule is applied only when none of the earlier rules are applicable.
{\bf (Step 0)} If $M$ has COP and $d \geq 0$ then Return ${\cal D}$.\\
{\bf (Step 1)} If $d < 0$ then Return 'NO'/* parameter budget exhausted */ \\
{\bf (Step 2)}(Branching Rule 1) If there exists three pairwise intersecting sets $S_1,S_2,S_3 \in {\cal S}(M)$ satisfying either of the following properties:\\
(H1) $S_1 \cap S_2 \cap S_3 = \emptyset$.\\
(H2) None of $S_1,S_2$ and $S_3$ is contained in the union of the other two. \\
%if ($S_1 \cap S_2 \cap S_3 = \emptyset$ or $S_1 \not \subseteq (S_2 \cup S_3)$ and $S_2 \not \subseteq (S_1 \cup S_3)$ and $S_3 \not \subseteq (S_1 \cup S_2)$ ) then  \\
%\{ \\
then branch into 3 instances ${\cal I}_i=(M_i,d_i)$ (where $i \in \{1,2,3\}$) \\
%For $i=1$ to $3$ /* Generate 3 recursive subproblems */\\%each $i \in \{S_1,S_2,S_3 \}$
Set ${\cal D}_i\leftarrow{\cal D} \cup \{r_i\}$ and $M_i\leftarrow M \setminus \{r_i\}$ \\
Update $d_i \leftarrow d-1$ /* Parameter drops by 1 */ \\
For some $i \in \{1,2,3\}$, if COS-R($M_i,{\cal D}_i ,d_i$) returns a solution ${\cal D}_i$, then Return ${\cal D}_i$, else Return 'NO'\\
%\indent \hspace{.5cm} Set ${\cal D}$= COS-R($M_i,{\cal D}_i ,d_i$) and Return ${\cal D}$ \\
/* {\em \footnotesize{Invariant: }} See Lemma \ref{two-columns} and Corollary \ref{max-cli} */\\
%For every maximal clique $Q$ in $G(M)$ that has no column $c_i$ in $M$ with $vert(c_i)=Q$, there exists two columns $c_i$ and $c_j$ in $M$ such that $Q \subseteq vert(c_i) \cup vert(c_j)$.  */\\
{\bf (Step 3)}(Branching Rule 2) If there exists two columns $c_i$ and $c_j$ in $M$ such that $G[vert(c_i) \cup vert(c_j)]$ has an induced cycle $C=\{v_1,v_2,v_3,v_4\}$,\\
then branch into 4 instances ${\cal I}_i=(M_i,d_i)$ (where $i \in \{1,2,3,4\}$) \\
%For $i=1$ to $4$ /* Generate 4 recursive subproblems */\\%each $t \in \{v_1,v_2,v_3,v_4 \}$ 
Set ${\cal D}_i\leftarrow{\cal D} \cup \{r_i\}$ and $M_i\leftarrow M \setminus \{r_i\}$ \\
Update $d_i \leftarrow d-1$ /* Parameter drops by 1 */ \\
For some $i \in \{1,2,3,4\}$, if COS-R($M_i,{\cal D}_i ,d_i$) returns a solution ${\cal D}_i$, then Return ${\cal D}_i$, else Return 'NO'\\
%\indent \hspace{.5cm} Set ${\cal D}$= COS-R($M_i,{\cal D}_i,d_i$) and Return ${\cal D}$ \\
/* {\em \footnotesize{Invariant: }} See Lemma \ref{rules12} */\\
%The graph $G[vert(c_i) \cup vert(c_j)]$ induced on any pair $c_i$ and $c_j$ of columns in $M$ is chordal. */\\
{\bf(Step 4)}(Branching Rule 3) If there is a maximal clique $Q$ such that there does not exist a column $c_l$ such that $vert(c_l)=Q$ then, let $Q'$ be a minimal subset of $Q$ with the property that there is no column $c_{l'}$ such that $Q \subseteq vert(c_{l'})$.\\
/*$Q'$ is well-defined as it is a subset of $Q$ and $Q$ itself is in two columns */ \\
Let $v_1$,$v_2$,$v_3$ be vertices in $Q'$, and let the corresponding rows be $r_1$,$r_2$,$r_3$ respectively.\\
%  If there is a set $Q'$ (with at least 3 vertices $v_1$, $v_2$ and $v_3$) that is a subset of a maximal clique $Q$ in $G(M)$ such that for each $v_i \in Q'$, there is a column $c_j$ in $M$ with $Q'\setminus \{v_i\} \subseteq vert(c_j)$ and there is no column $c_l$ in $M$ such that $Q' \subseteq vert(c_l)$\\
then branch into 3 instances ${\cal I}_i=(M_i,d_i)$ (where $i \in \{1,2,3\}$) \\
Set ${\cal D}_i\leftarrow{\cal D} \cup \{r_i\}$ and $M_i\leftarrow M \setminus \{r_i\}$ \\
Update $d_i \leftarrow d-1$ /* Parameter drops by 1 */ \\
For some $i \in \{1,2,3\}$, if COS-R($M_i,{\cal D}_i ,d_i$) returns a solution ${\cal D}_i$, then Return ${\cal D}_i$, else Return 'NO'\\
%\indent \hspace{.5cm} Set ${\cal D}$= COS-R($M_i,{\cal D}_i,d-|{\cal R}_i|$) and Return ${\cal D}$ \\
/* {\em \footnotesize{Invariant: }} See Lemma \ref{rules123} */\\
{\bf (Step 5)}(Interval Deletion) $V'$=Interval-Deletion($G(\oset{$\thicksim$}M),d$). \\
{\bf (Step 6)}If Interval-Deletion returns `NO' then Return.\\
\indent \hspace{1.5cm} Otherwise, Return the set ${\cal D}= {\cal D} \cup \{r_i \in {\cal R}(M)\mid v_i \in V'\}$. %/* ${\cal D}$ is a set of at most $d$ rows such that $M \setminus {\cal D}$ has COP */
}
} \\
At each leaf in the recursion tree, an interval deletion problem is solved.  Each node in the recursion tree has at most 4 subproblems, and therefore, the tree has at most $4^d$ leaves, and then using the recent FPT algorithm for Interval Deletion \cite{CM12},
we get an overall running time of $O^*(10^d)$ for our algorithm.
Recall that, for a matrix $M$, the derived graph is denoted by $G(M)$ and its set system is denoted by ${\cal S}(M)=\{S_1,\ldots,S_m\}$.
The recursive function COS-R is called initially with the input matrix $M$, the initial solution set ${\cal D}=\emptyset$ and the parameter $d$ as inputs.  
 It either returns a set ${\cal D}$ of at most $d$ rows such that $M \setminus {\cal D}$ has COP or returns 'NO'. COS-R makes a call to the function Interval-Deletion$(G,d)$ which {\em either} returns a set of vertices $X$ such that $|X| \leq d$, and  $G \setminus X$ is an interval graph or returns 'NO'.\\ \\
\noindent 
{\bf Correctness of the Algorithm: }
We prove the correctness of the algorithm by proving invariants that hold at the end of each branching rule.  
\begin{lemma}
\label{rule1}
Let $M$ be a matrix for which branching rule 1 applies, and sets $S_1, S_2, S_3$ violate at least one of the two conditions checked in rule 1. Then, any solution $\cal D$ of $d$-COS-R includes at least one of the corresponding rows $r_1, r_2, r_3$.
\end{lemma}
\begin{proof}
The proof follows from Lemma \ref{lem-helly}.
\end{proof}
\noindent
{\bf Branching Rule 1: } 
%We first apply branching rule 1 to ensure that Helly property is satisfied for every triple of sets in ${\cal S}(M)$. The correctness of this rule follows from Corollary \ref{lem-helly}. 
To understand the effect of Branching Rule 1, consider this example of the matrices $M_1=\left( \begin{smallmatrix}1 & 1 & 1 & 1 & 1 & 0 & 0 & 0 \\0 & 1 & 0 & 0 & 1 & 1 & 1 & 0 \\ 1 & 0 & 1 & 1 & 0 & 1 & 0 & 1 \end{smallmatrix} \right)$ and $M_2=\left( \begin{smallmatrix}1 & 0 & 1 & 0 & 1 & 1 & 0 & 0\\0 & 1 & 0 & 1 & 0 & 1 & 1 & 0 \\ 1 & 0 & 0 & 0 & 0 & 1 & 0 & 1 \end{smallmatrix} \right)$, both do not have COP.   In $M_1$ and $M_2$, the sets corresponding to the rows are pairwise intersecting. However, in $M_1$ the sets do not have a common element while in $M_2$, none of them is contained in the union of other two. The following lemma formalizes the crucial property satisfied by matrices for which branching rule 1 is not applicable.
\begin{lemma}
\label{two-columns}
Let $M$ be a matrix on which branching rule 1 is not applicable. Then, for every maximal clique $Q$ in $G(M)$, there are at most two columns $c_i$ and $c_j$
such that $Q \subseteq vert(c_i) \cup vert(c_j)$. 
% is such that $ver(c_i)=QFor every maximal clique $Q$ in $G(M)$, there exists a set $T$ of at most two columns such that $\bigcup_{c_i \in T} vert(c_i)=Q$. 
\end{lemma}
\begin{proof}Assume on the contrary that $Q$ is a maximal clique in $G(M)$ and $T$ is a minimum set of columns such that $Q \subseteq \bigcup_{c_i \in T} vert(c_i)$ with $|T| \geq 3$. Consider the submatrix $N$ with ${\cal R}(N)=\{ r_i \in {\cal R}(M) \mid v_i \in Q \}$. Consider any 3 columns $c_1, c_2, c_3$ from $T$.  Since $T$ is a minimum set of columns whose vertices contain $Q$ in $G(M)$, it follows that there are 3 vertices $v_1, v_2, v_3 \in Q$ such that the corresponding rows along with the colums $c_1, c_2, c_3$ form an identity submatrix which can be visualized as $\left( \begin{smallmatrix} \cdots & 1 & \cdots &0 & \cdots & 0 & \cdots\\ \cdots & 0 & \cdots &1 & \cdots & 0 & \cdots\\\cdots & 0 & \cdots &0 & \cdots & 1 & \cdots\end{smallmatrix} \right)$.
% as a sub-matrix induced on rows, say, $r_1$, $r_2$ and $r_3$.
Thus each of the  sets $S_1,S_2$ and $S_3$, corresponding to $r_1, r_2, \text{ and } r_3$, has an element that is not present in the other two. Therefore, none of $S_1, S_2$ and $S_3$, is contained in the union of the other two, therefore branching rule 1 would have been applied. This is a contradiction to the hypothesis in the lemma that branching rule 1 is not applicable.  \qed
\end{proof}
\begin{corollary}
\label{max-cli}
Every maximal clique in $G(M)$ is a maximal clique in $G[vert(c_i) \cup vert(c_j)]$ for some pair of columns $c_i,c_j $ in $M$.
\end{corollary}
\noindent An example is shown in Figure \ref{fig1}. The maximal cliques $Q_1$ and $Q_2$ in $G(M)$ are such that $Q_1$=$\{v_1,v_2,v_4,v_5,v_6\}$ = $vert(c_1) \cup vert(c_5)$ and $Q_2=\{v_2,v_3,v_4,v_5,v_6\}$ = $vert(c_3) \cup vert(c_4)$. It is also clear from the figure that no five clique is {\em present in a column}.\\ 
\begin{figure}[h]
\centering
\includegraphics[scale=.26]{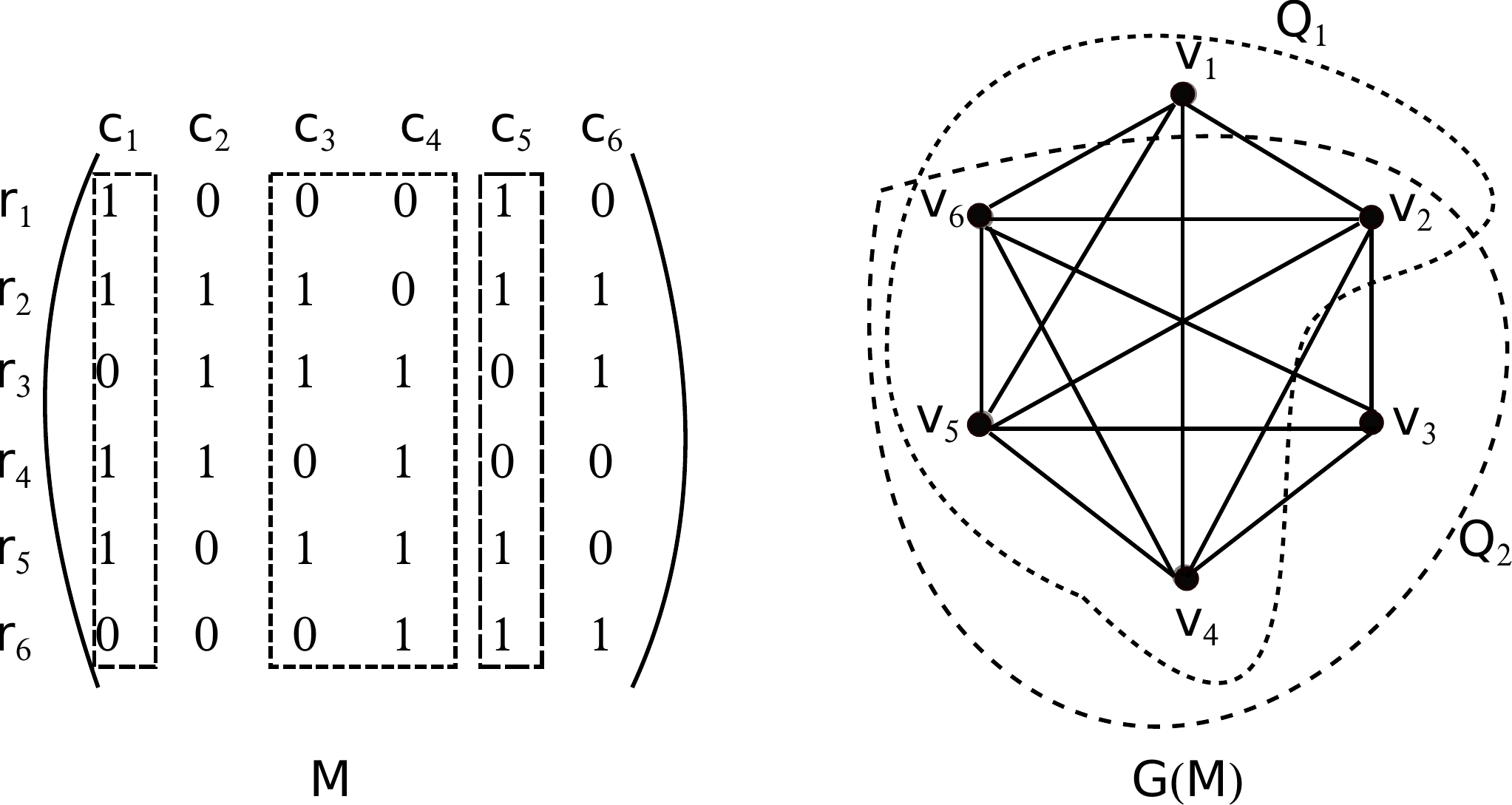}
 \caption{Maximal cliques $Q_1$ and $Q_2$ in $G(M)$ \label{fig1}}
\end{figure}

%Observe that the naive approach of enumerating the maximal cliques of $G(M)$ and checking if each clique satisfies the required property is not guaranteed to run in polynomial time. However, f
\noindent\textbf{Branching Rule 2: }Let $M$ be a matrix on which branching rule 1 is not applicable. Now for each maximal clique $Q$ in $G(M)$ for which there does not exist  a column $c_k$ in $M$ such that $Q=vert(c_k)$, we need to branch according to Lemma \ref{symm-diff}.  An important question here is that how do we check if there are such cliques.  From Corollary \ref{max-cli}, we know that the maximal cliques of $G(M)$ can be enumerated by enumerating the maximal cliques of $G[vert(c_i) \cup vert(c_j)]$ for each pair of columns $c_i$ and $c_j $ in $M$.  However, there could be an exponential number of 
maximal cliques in  $G[vert(c_i) \cup vert(c_j)]$.   We handle this problem  by checking if $G[vert(c_i) \cup vert(c_j)]$ is non-chordal.  If it is chordal, then there are only a polynomial number of maximal cliques, and it is easy to enumerate each maximal clique, and branch as suggested
in Lemma \ref{symm-diff}.  If $G[vert(c_i) \cup vert(c_j)]$ is not chordal, there is a chordless cycle of length more than 3, and we show in the following lemma that such a chordless cycle can only be of length 4.  However, from Lemma \ref{cop-ig} and Theorem \ref{char-int}, it follows that any induced cycle of length 4 is forbidden in $G(M)$. Therefore, by our branching rule 2, we guarantee that $G[vert(c_i) \cup vert(c_j)]$ does not have an induced cycle of length 4. We now show in the following two lemmas that this guarantees that $G[vert(c_i) \cup vert(c_j)]$ is chordal. 
\begin{lemma}
\label{4-chord}
Let $c_p$ and $c_q$ be two columns in $M$ on which branching rule 1 is not applicable. Then, every induced cycle in $G'=G[vert(c_p) \cup vert(c_q)]$ is of length at most 4.  
\end{lemma}
\begin{proof}
Any induced cycle in $G'$ can have at most two vertices from $vert(c_p)$ and $vert(c_q)$ each, as they both induce cliques in $G'$.  Therefore, any induced cycle can be of length at most 4. \qed
\end{proof}
\begin{lemma}
\label{corr-rule2}
Let $M$ be a matrix on which branching rule 1 is not applicable.  For two columns $c_p$ and $c_q$ in $M$, let $C$ be an induced cycle of four vertices 
such that $C$ is in $G[vert(c_p) \cup vert(c_q)]$.   Then, any solution $\cal D$ of $d$-COS-R must include  at least one of the four rows corresponding to the four vertices in $C$.
\end{lemma} 
\begin{proof}
Let ${\cal D}$ be a solution and let $M'=M \setminus {\cal D}$.  If $C$ is in $G(M')$, then it means that $G[M']$, an interval graph by Lemma \ref{cop-ig},
has an induced cycle of length 4, which is a contradiction to Theorem \ref{char-int}, which characterizes interval graphs as a subclass of graphs without induced  cycles on 4 vertices.
Hence the lemma.   \qed
\end{proof}
% \begin{figure}[h]
%\centering
%\includegraphics[scale=.4]{induced-4cycle.pdf}
%\caption{Induced cycles in $G[vert(c_2) \cup vert(c_4)]$ }
%\end{figure}
\begin{lemma}
\label{rules12}
Let $M$ be a matrix on which branching rule 1 and branching rule 2 are not applicable.  Then the following are true:\\
(1) For each maximal clique $Q$ in $G(M)$, there exists at most two columns $c_p$ and $c_q$ in $M$ such that $Q \subseteq vert(c_p) \cup vert(c_q)$\\
(2) For each pair of columns $c_p$ and $c_q$ in $M$, $G[vert(c_p) \cup vert(c_q)]$ is chordal.  
\end{lemma}
\begin{proof}
If either of the two conditions are not true, then it would contradict the premise that the two branching rules are not applicable.  \qed
\end{proof}
{\bf Branching Rule 3: }After applying branching rule 2, $G[vert(c_p) \cup vert(c_q)]$ is chordal for each pair of columns $c_p,c_q$ in $M$. It is known that the maximal cliques of a chordal graph can be enumerated in linear time \cite{Gol04}. So, we enumerate the maximal cliques of $G[vert(c_i) \cup vert(c_j)]$ for each pair of columns $c_i$ and $c_j$ in $M$. From Corollary \ref{max-cli}, this enumeration is guaranteed to list all the maximal cliques of $G(M)$. For each maximal clique $Q$ in this enumeration,  if there is no column $c_k$ such that $vert(c_k) = Q$, then we identify two columns $c_p$ and $c_q$ such that $Q \subseteq vert(c_p) \cup vert(c_q)$, and apply branching rule 3.  The following lemma proves that Branching Rule 3 is necessary.
\begin{lemma}
\label{symm-diff}
Let $M$ be a matrix on which branching rule 1 is not applicable. Let $Q$ be a maximal clique in $G(M)$ such that there is no column $c_l$ such that $vert(c_l) = Q$. Let $Q'$ be a minimal subset of $Q$ that has no column $c_{l'}$ such that $Q' \subseteq vert(c_{l'})$. Let $v_1$,$v_2$ and $v_3$ be any three vertices in $Q'$, and let $r_1$,$r_2$, $r_3$ respectively be the corresponding rows. Then, any solution $\cal D$ of 
$d$-COS-R must include at least one of $r_1$, $r_2$ and $r_3$.
\end{lemma} 
\begin{proof}
We prove this by contradiction. Suppose there exists a solution ${\cal D}$ that contains none of $r_1$, $r_2$ and $r_3$. Let $M'=M \setminus {\cal D}$ be the matrix with COP. Since there is no column $c_l$ in $M$ such that $Q' \subseteq vert(c_l)$ and $Q'$ is an inclusion minimal with this property, it follows that there exists distinct columns $c_1$, $c_2$ and $c_3$ such that $Q'\setminus \{v_i\} \subseteq vert(c_i)$ for each $i \in \{1,2,3\}$. Now, it follows that the rows $r_1$, $r_2$ and $r_3$ along with the colums $c_1, c_2, c_3$ form a submatrix of $M'$ which can be visualized as $\left( \begin{smallmatrix} \cdots & 0 & \cdots &1 & \cdots & 1 & \cdots\\ \cdots & 1 & \cdots &0 & \cdots & 1 & \cdots\\\cdots & 1 & \cdots &1 & \cdots & 0 & \cdots\end{smallmatrix} \right)$. This submatrix is forbidden for any matrix with COP \cite{AT72}. This is a contradiction to the fact that $M'$ has COP. Therefore, our assumption is wrong, and hence the lemma is proved. \qed
%Observe that Q' is of size at least 4. For each clique Q'' of size at most 2, there is a column c_r such that Q'' \subseteq vert(c_r). Also, for a clique Q''=\{v_a,v_b,v_c\} of size 3, since the corresponding sets S_a, S_b and S_c are pairwise intersecting, there exists x \in (S_a \cap S_b \cap S_c). Note that x is well-defined as (S_a \cap S_b \cap S_c) \neq \emptyset by Helly property. Thus, c_x is a column such that Q'' \subseteq vert(c_x). Therefore, Q' indeed has at least four vertices and hence v_1, v_2 and v_3 are well-defined.
\end{proof}
\begin{lemma}
\label{rules123}
Let $M$ be a matrix for which branching rule 1, branching rule 2, and branching rule 3 are not applicable.  Then, for each maximal clique $Q$ in $G(M)$, there exists a column $c_p$ such that $Q = vert(c_p)$.  Further, $\oset{$\thicksim$}M$ is the clique matrix of $G(\oset{$\thicksim$}M)$.
\end{lemma}
\begin{proof}
The proof of this lemma, too, is by contradiction. If $Q$ is a maximal clique such that there is no column $c_p$ such that $vert(c_p) = Q$, then
since branching rule 1 and branching rule 2 are not applicable, by Lemma \ref{rules12}, it follows that there exist columns $c_p$ and $c_q$ such that $Q \subseteq vert(c_p) \cup vert(c_q)$.  This implies that branching rule 3 is applicable for $M$, and this contradicts the premise of the lemma.  
Therefore, our assumption is wrong, and the first part of the lemma is proved.   To prove the second part of the lemma- Any column $c_k$ whose vertices $vert(c_k)$ is not a maximal clique in $G(M)$ becomes a maximal clique in $G(\oset{$\thicksim$}M)$. This is because $G(\oset{$\thicksim$}M)$ can be viewed as a graph obtained from $G(M)$ by adding a new vertex for each clique $vert(c_k), 1 \leq k \leq n$, and making this vertex adjacent to all the vertices in $vert(c_k)$.   Further, if $vert(c_k)$ is a maximal clique in $G(M)$, then in $G(\oset{$\thicksim$}M)$, $vert(c_k)$  is a maximal clique with one additional vertex.  This completes the proof that  $\oset{$\thicksim$}M$ is the clique matrix of $G(\oset{$\thicksim$}M)$.
Hence the lemma. \qed
\end{proof}
Now we show that, solving $d$-COS-R on  $\oset{$\thicksim$}M$ is equivalent to solving the Interval-Deletion problem on the graph $G(\oset{$\thicksim$}M)$.
\begin{theorem}
\label{interdel}
Let $\oset{$\thicksim$}M$ be the clique matrix of $G(\oset{$\thicksim$}M)$.
Given $\oset{$\thicksim$}M$ and integer $d \geq 0$, there exists a set of rows ${\cal D}$ such that $|{\cal D}| \leq d$ and 
$\oset{$\thicksim$}M \setminus {\cal D}$ has COP if and only if $G(\oset{$\thicksim$}M)$ has a set of vertices $V'$ such that $|V'| \leq d$ and $G(\oset{$\thicksim$}M) \setminus V'$ is an interval graph.
\end{theorem}
\begin{proof}
Let ${\cal D}$ be a set of rows in $\oset{$\thicksim$}M$, and let $V'$ be the corresponding vertices in $G(\oset{$\thicksim$}M)$.   
From Lemma \ref{cop-ig} it follows that $\oset{$\thicksim$}M \setminus {\cal D}$ has COP implies $G(\oset{$\thicksim$}M \setminus {\cal D})$ is an 
interval graph. Further, $G(\oset{$\thicksim$}M \setminus {\cal D})$ is basically the graph obtained by removing $V'$ from $G(\oset{$\thicksim$}M)$.  This completes the forward direction of the claim.  In the reverse direction, let $V'$ be a minimal set of vertices such that $G(\oset{$\thicksim$}M) \setminus V'$ is an interval graph.  Due to the minimality of $V'$ observe that the vertices in $G(\oset{$\thicksim$}M)$ which correspond to
the rows of the identity matrix added to $M$ are not elements of $V'$.  Let ${\cal D}$ be the set of rows in $\oset{$\thicksim$}M$ corresponding to $V'$.   Note that the columns of $\oset{$\thicksim$}M \setminus {\cal D}$ are exactly the maximal cliques of $G(\oset{$\thicksim$}M) \setminus V'$.
Therefore, $\oset{$\thicksim$}M \setminus {\cal D}$ is the clique matrix of $G(\oset{$\thicksim$}M) \setminus V'$.  Since $G(\oset{$\thicksim$}M) \setminus V'$ is an interval graph, it follows from Theorem \ref{FG} that $\oset{$\thicksim$}M \setminus {\cal D}$ has COP.  Hence the theorem is proved. \qed
\end{proof}
We now show that, the recursive function $d$-COS-R correctly decides whether a given matrix $M$ has a set of at most $d$-rows whose removal results in a matrix with COP.
\begin{theorem}
\label{our-result}
%\label{main-result}
Given an instance  $(M,d)$ of $d$-COS-R, the function call COS-R$(M,\emptyset,d)$ correctly decides in $O^*(10^d)$ time if there exists a set $\cal D$ of at most $d$ rows such that $M \setminus {\cal D}$ has COP. 
\end{theorem}
\begin{proof}
 Let ${\cal D}$ be a solution of size at most $d$.  From the Lemma
\ref{rule1}, Lemma \ref{corr-rule2}, and Lemma \ref{symm-diff}, in each recursive subproblem, one of the rows to be added in the solution is an element of ${\cal D}$.   Let $(M',d')$ be the instance of $d$-COS-R at a leaf node in the recursion, where this leaf node is one at which none of the first three branching rules apply, and each of recursive choices of rows to be added into the solution, in the computation starting at COS-R$(M,\emptyset,d)$ is selected from ${\cal D}$.   Let ${\cal D'} \subseteq {\cal D}$ be the set of rows that have been added to the solution in recursive calls upto the leaf node at which $(M',d')$ is an instance of $d$-COS-R, and let ${\cal D''} = {\cal D} \setminus {\cal D'}$.  From Corollary \ref{corr-cop}, 
 $M' \setminus {\cal D''}$ has COP if and only if $\oset{$\thicksim$}M' \setminus {\cal D''}$ has COP.  Further, from Lemma \ref{rules123}, $\oset{$\thicksim$}M'$ is the clique matrix of $G(\oset{$\thicksim$}M')$ at the leaf node in the recursion tree.  Therefore, it follows that $\oset{$\thicksim$}M' \setminus {\cal D''}$ has COP, and from Theorem \ref{interdel}, that $G(\oset{$\thicksim$}M') \setminus V'$ is an interval graph.   Therefore, from \cite{CM12}, it follows that Interval-Deletion$(G(\oset{$\thicksim$}M'),d')$ will return a set of at most $d'$ vertices whose removal from $G(\oset{$\thicksim$}M')$ guarantees that the resulting graph is an interval graph.  This proves that if there is a solution ${\cal D}$ to $(M,d)$, then COS-R$(M,\emptyset,d)$ will return a solution of size at most $d$.  It is also clear that if there is no solution ${\cal D}$ of size at most $d$, the algorithm will not find one. 
 
%$d$-COS-R
\noindent
In each of the recursive subproblems generated by branching rules 1, 2, and 3, the parameter reduces by at least 1.  Further, in each level of recursion,  at most four recursive calls are made, in branching rules 1, 2, and 3.  Therefore, in the recursion tree obtained by performing the 3 branching rules, there are at most $4^{d-d'}$ leaves at depth $d-d'$.    At a leaf node, in which the problem is $(M',d')$, Interval-Deletion returns an answer in at most $O^*(10^{d'})$ time.  Further,  the checks made at each level of recursion takes only polynomial time.  Therefore, this bounds the total running time of the algorithm by $O^*(10^d)$.    Hence the theorem. \qed
\end{proof}
%\begin{proof}
%Lemmas $7,8,9$ achieves the correctness of the branching rules $1-3$. After the branching rules, we obtain a matrix $M$ in which for every maximal clique $Q$ in $G(M)$ there exists a column $c_k$ in $M$ such that $supp(c_k)=Q$. Then, if $M$ does not have COP, we employ the interval deletion algorithm at the leaves of the search tree built by the branching rules. It is known that the interval deletion algorithm runs in time $O^*({10^d})$ \cite{CM12}. Since, there are at most 4 branches that is performed in the preprocessing rules, our algorithm runs in time $O^*({10^d})$. \qed
%\end{proof}
%\begin{theorem}
%Given a bipartite graph, there is an $O^*(10^d)$ algorithm for deciding if there is a set of at most $d$ vertices whose deletion makes the graph $G$ convex bipartite. 
%\end{theorem}
%Maximum matching in convex bipartite  graphs have interesting industrial applications \cite{LP81}.
\section{Concluding Remarks}
Using our algorithm for $d$-COS-R, we observe that the Convex Bipartite Deletion problem is FPT. Let $G=(V_1,V_2,E)$ be a bipartite graph with $V_1=\{x_1,\ldots,x_m\}$ and $V_2=\{y_1,\ldots,y_m\}$. Let $M$ be the half adjacency matrix of $G$. That is, $M_{ij}=1$ if and only if $\{x_i,y_j\} \in E$. $G$ is {\em convex} bipartite graph if and only if $M$ has COP \cite{AT72,MD08}. The Convex Bipartite Deletion problem is defined as follows.\\

\noindent \fbox{
\parbox{16cm}{
%\parbox{11.7cm}{
\indent \hspace{5cm} \textbf{Convex Bipartite Deletion} \\
\textbf{Input:} A bipartite graph $G=(V_1,V_2,E)$, $|V_1|=m$, $|V_2|=n$ and $d \ge 1$ \\
\textbf{Parameter:} $d$ \\
\textbf{Question:} Does there exist a set $D \subset V_1$ with $|D| \le d$ such that $G[V_1\setminus D,V_2]$ is a convex bipartite graph?
}
}\\ \\
\noindent This problem is known to be NP-complete from \cite{LY80}. However, from Theorem \ref{our-result}, the COS-R algorithm in Section \ref{our-algo} can be used to solve the problem in $O^*(10^d)$ time. Here, the inputs to the algorithm are the half adjacency matrix $M$ of $G$ and the parameter $d$. The algorithm returns a set ${\cal D}$ of at most $d$ rows (if one exists) such that $G[V_1\setminus D,V_2]$ is convex bipartite where $D$ is the subset of vertices of $V_1$ corresponding to ${\cal D}$.
%\indent A problem related to $d$-COS-R is to find a minimum number of rows and columns whose deletion from the input matrix results in a matrix with COP. It remains to explore if our branching rules can be tweaked to solve this problem. Also, extending our approach to determine a minimum set of rows to delete so that the resultant matrix has circular ones property is an interesting direction of research.
%for deleting minimum number of rows/columns/rows and columns of the matrix, such that the resulting submatrix 

\bibliography{refs}
\bibliographystyle{splncs03} 
\end{document}